\documentclass[AMA,STIX1COL]{WileyNJD-v2}

\articletype{Article Type}%

\newcommand{\EulerVec}[2][]{\mathbf{{#2}}}
\newcommand{\braket}[3][]{ \langle {#2} , {#3} \rangle }

\newcommand{\Amat}[2][]{ \mathbf{#2} }
\newcommand{\Avec}[2][]{ \mathbf{#2} }

\makeatletter
\providecommand*{\diff}%
{\@ifnextchar^{\DIfF}{\DIfF^{}}}
\def\DIfF^#1{%
\mathop{\mathrm{\mathstrut d}}%
\nolimits^{#1}\gobblespace}
\def\gobblespace{%
\futurelet\diffarg\opspace}
\def\opspace{%
\let\DiffSpace\!%
\ifx\diffarg(%
\let\DiffSpace\relax
\else
\ifx\diffarg[%
\let\DiffSpace\relax
\else
\ifx\diffarg\{%
\let\DiffSpace\relax
\fi\fi\fi\DiffSpace}

\newcommand{\deriv}[3][]{%
\frac{\diff^{#1}#2}{\diff #3^{#1}}}

\newcommand{\AAA}{\mathbf{a}}
\newcommand{\EE}{\mathbf{e}}

\newcommand{\JJ}{\mathbf{j}}
\newcommand{\nn}{\mathbf{n}}

\newcommand{\hsol}{H_L(\Omega)}
\newcommand{\hport}{H_L(\Omega,\Sigma)}
\newcommand{\R}{\mathbb{R}}
\newcommand{\dive}{\rm{div}}
\newcommand{\Lplus}{L^\infty_+(\Omega)}
\newcommand{\dd}{\mathrm{d}}
\newtheorem{thm}{Theorem}[section]
\newtheorem{rem}[thm]{Remark}
\newtheorem{prop}[thm]{Proposition}
\newcommand{\OA}{\mathcal{A}}
\newcommand{\OP}[1]{\mathcal{#1}}
\DeclareMathOperator{\im}{im}
\newcommand{\jd}{\alpha}
\newcommand{\JD}{\mathbf{\jd}}

\usepackage{subfig}
\usepackage{comment}
\usepackage{steinmetz}

\begin{document}

\title{Modal Decomposition in Numerical Computation of Eddy Current Transients}


\author[1]{Salvatore Ventre}

\author[2]{Andrea Chiariello}

\author[3]{Nicola Isernia}

\author[1]{Vincenzo Mottola}

\author[1,4]{Antonello Tamburrino}

\authormark{S. Ventre \textsc{et al}}

\address[1]{\orgdiv{
Dipartimento di Ingegneria Elettrica e dell’Informazione}, \orgname{Universita di Cassino e del Lazio Meridionale}, \orgaddress{Via Gaetano di Biasio 43, 03043 Cassino, \country{Italy}}}

\address[2]{\orgdiv{Dipartimento di Ingegneria}, \orgname{Università degli Studi della Campania Luigi Vanvitelli}, \orgaddress{via Roma 29, 81031 Aversa, \country{Italy}}}

\address[3]{\orgdiv{Dipartimento di Ingegneria Elettrica e delle Tecnologie dell'Informazione}, \orgname{Università degli Studi di Napoli Federico II}, \orgaddress{Via Claudio 21, 80125 Napoli, \country{Italy}}}

\address[4]{\orgdiv{Department of Electrical and Computer Engineering}, \orgname{Michigan State University}, \orgaddress{East Lansing (MI-48824), \country{USA}}}

\corres{*Andrea Chiariello, \email{andreagaetano.chiariello@unicampania.it}}


\abstract[Summary]{A methodology to reduce the computational cost of time domain computations of eddy currents problems is proposed. It is based on the modal decomposition of the current density. In particular, a convenient strategy to perform the modal decomposition even in presence of injected currents into the electrodes of a conducting domain is proposed and implemented in a parallel computing environment. Using a theta-method integration algorithm, the performances of the the proposed approach are compared against those of a classical method based on the Cholesky factorization, for a case of interest from eddy current nondestructive testing. For this large eddy current problem (number of unknowns greater than 100k, number of time steps of interest equal to 100k) the proposed solution method is shown to be much faster than those based on standard time integration schemes.}



\maketitle


\section*{List of symbols and conventions}

\begin{table}[h]
    \centering
    \begin{tabular}{c|c}
        $\Avec[]{A}$ & discrete vector (1D array) \\
        $\Amat[]{A}$      & discrete matrix (2D array) \\
        $\EulerVec[]{a}$(x) & vector field in Euclidean three-dimensional space \\
        $\braket[]{A}{B}$ & Euclidean scalar product between scalar or vector fields in $L_2(\Omega)$ \\
        $\braket[]{A}{B}_{\Sigma_k}$ & scalar product on $\Sigma_k$\\
        $ \OP{A} $ & Linear operator $H_L(\Omega, \partial \Omega) \rightarrow L_2(\Omega, \mathbb{R}^3)$ defined by (\ref{operat}) \\
        $ \OP{I} $ & Linear operator $H_L(\Omega, \Sigma) \rightarrow \mathbb{R}^{N_e}$ defined by (\ref{eq:operator_electrodes}) \\
    \end{tabular}
    \caption{Notation adopted throughout the article.}
    \label{tab:my_label}
\end{table}

\section{Introduction}\label{sec:introduction}

In a variety of applications, ranging from the design of magnets of fusion devices or accelerators to non-destructive testing, it is necessary to perform long time-domain eddy currents simulations over very large or densely discretized conducting domains. The computation of eddy currents in bulk conductors for assigned induced or applied voltages from external sources is a standard engineering problem, which is usually addressed by the means of integral formulations and finite element techniques\cite{Albanese1988}. Mathematically, the problem is equivalent to that of solving a linear system of Ordinary Differential Equations, in presence of eventual holonomic constraints. In this manuscript we quantify when and why it is convenient to perform a modal decomposition of the original problem, rather than resorting to classic numerical integration schemes based on the Cholesky decomposition.

\section{Modal Decomposition of the Eddy Current Density }\label{sec:maths}
\subsection{Notations and functional spaces}
Throughout this paper, $\Omega$ is the region occupied by the conducting material. We assume $\Omega\subset\R^3$ to be an open bounded domain with a Lipschitz boundary and outer unit normal $\hat {\bf n}$.
We denote by $V$ and $S$ the $3$-dimensional and the bi-dimensional Hausdorff measure in $\R^3$, respectively and by $\langle\cdot,\cdot\rangle$ the usual $L^2$-integral product on $\Omega$. 

{Hereafter we refer to the following functional spaces
\begin{align}
\label{eq:L_infinity}
L^\infty_+(\Omega) & :=\{\theta\in L^\infty(\Omega) \ | \ \theta\geq c_0\ \mathrm{a.e.\ in} \ \Omega \ \mathrm{for}\ c_0>0\},\\
H_{\mathrm{div}}(\Omega) & :=\{ {\bf v} \in L^2(\Omega;\R^3) \ |\ \mathrm{div} ({\bf v}) \in L^2(\Omega)\},\\
H_{\mathrm{curl}}(\Omega) & :=\{ {\bf v} \in L^2(\Omega;\R^3) \ |\ \mathrm{curl} ({\bf v}) \in L^2(\Omega)\},\\
\hsol & :=\left\{  \mathbf{v}\in H_{\rm div}(\Omega)\   | \ \dive (\mathbf{v})=0\ \mathrm{in}\ \Omega,\ \mathbf{v}\cdot\hat{\bf n}=0\ \mathrm{on}\ \partial \Omega\right\},\\
\label{eq:HL}
\hport & :=\left\{  \mathbf{v}\in H_{\rm div}(\Omega)\   | \ \dive (\mathbf{v})=0\ \mathrm{in}\ \Omega,\ \mathbf{v}\cdot\hat{\bf n}=0\ \mathrm{on}\ \partial \Omega \backslash \Sigma \right\},
\end{align}
and to the derived spaces $L^2(0,T;H_L(\Omega))$, $L^2(0,T;\hport)$ and $L^2(0,T;H_{\rm curl}(\Omega))$, for any $0<T< +\infty$.

It is worth noting that $\hsol$ is required to account for vector fields that cannot enter/leave the domain $\Omega$, whereas $H_L(\Omega,\partial \Omega)$ is required to account for vector fields that can enter/leave the domain $\Omega$ through the entire $\partial \Omega$. Finally, $\hport$ is required to account for vector fields that can leave the domain $\Omega$ \emph{only} through surface $\Sigma$. Moreover,
\begin{equation*}
    \hsol = H_L(\Omega,\emptyset) \subset \hport \subset H_L(\Omega,\partial \Omega).
\end{equation*}

We refer to \cite{bossavit1998computational,bossavit1981numerical} for $\hsol$ and a general discussion on functional spaces related to Maxwell equations.

}

\subsection{Mathematical Model for the Eddy Current Problem}
The eddy current problem, corresponding to Maxwell's equation in the Magneto-Quasi-Stationary limit (see \cite{haus1989electromagnetic}), can be cast in several different manners. Among them, an interesting formulation is that proposed by A. Bossavit (see \cite{bossavit1981numerical}) consisting in the following integral formulation in the weak form:
\begin{eqnarray}
\label{weak_form}
\left\langle \eta \JJ, \mathbf{w}\right\rangle + \langle\partial_{t}
\OA\JJ,\mathbf{w}\rangle=-\langle\partial_{t}\AAA_S, \mathbf{w}\rangle \quad \forall\mathbf{w}\in H_L(\Omega),
\end{eqnarray}
where $\JJ \in L^2(0,T;H_L(\Omega))$ is the unknown induced current density, $\Omega$ is the region occupied by the conducting material, $\OA$ is the compact, self-adjoint, positive definite operator:
\begin{eqnarray}
\label{operat}
\OA:\mathbf{v} \in H_L(\Omega,\partial \Omega) \mapsto\frac{\mu_0}{4\pi}\int_\Omega\frac{\mathbf{v}(x^\prime) }{\vert\vert x-x^\prime\vert\vert}\ \textrm{d}V(x')\in L^2(\Omega;\R^3),
\end{eqnarray}
$\eta \in \Lplus$ is the electrical resistivity, $\mu_0$ is the free space magnetic permeability, $\AAA_S \in L^2(0,T;H_{\rm curl}(\Omega))$ is the vector potential produced by the prescribed source current $\JJ_S\in L^2(0,T;H_{L}(\Omega_S))$, $\Omega_S$ is a bounded open set with Lipschitz boundary that is electrically insulated from $\Omega$.

\begin{rem}
The integral formulation of \eqref{weak_form} can be easily understood by combining (i) the constitutive relationship $\EE = \eta \JJ$, (ii) the electric field expressed in terms of vector and scalar potential as $\EE = - \partial_t \AAA - \partial_t \AAA_S -\nabla \phi$, being $\AAA$ the vector potential due to the induced current density $\JJ$ and $\phi$ the total scalar potential, and (iii) the vector potential due to the eddy current density expressed as $\AAA = \OA \JJ$:
\begin{equation}
\label{strong_form}
\eta \JJ + \partial_{t}
\OA\JJ + \nabla \phi=\partial_{t}{\bf a}_S \quad \text{ in } \ \Omega,
\end{equation}
The vector potential in \eqref{operat} is selected according the Coulomb's gauge. The Magneto-Quasi-Stationary approximation dictates the choice of the functional space ($\hsol$) and causes propagation to be neglected in the kernel of the operator $\OA$.
\end{rem}

It is worth noting that \eqref{weak_form} models the physical situation when $\partial \Omega$, the boundary of the conducting domain, is in contact with an electrical insulator, indeed $\JJ \cdot \hat{\nn} = 0 \text{ on } \partial \Omega$ since $\JJ \in \hsol$. A situation of relevant interest, which is the main focus of this contribution, is when electrical currents can flow through $\partial \Omega$ via a set of $N_e$ boundary electrodes. Let $\Sigma_k \subset \partial \Omega$, $k=1,\cdots,N_e$ be the surface patch representing the domain of the $k-$th electrode, then $\JJ \in L^2(0,T;\hport)$ where $\Sigma = \bigcup_{k=1}^{N_e} \Sigma_k$. If, in addition, (i) the electrodes are made by a perfect conducting material and (ii) the net current through each electrode is imposed and equal to $i_k(t)$, the equivalent of \eqref{weak_form}, obtained from \eqref{strong_form} written in the weak form, is
\begin{align}
\label{weak_form_Ik_a}
\left\langle \eta \JJ, \mathbf{w}\right\rangle + \langle\partial_{t}
\OA \JJ,\mathbf{w}\rangle + \sum_{k=1}^{N_e} \phi_k  \langle 1,\mathbf{w} \cdot \hat{\bf n} \rangle_{\Sigma_k} & = -\langle\partial_{t} \AAA_S, \mathbf{w}\rangle \quad \forall\mathbf{w}\in \hport\\
\label{weak_form_Ik_b}
\langle 1,\JJ \cdot \hat{\bf n} \rangle_{\Sigma_k} & = i_k^D \quad k=1,\ldots,{N_e},
\end{align}
being $\JJ \in L^2(0,T;\hport)$. In \eqref{weak_form_Ik_a} it has been exploited that (i) implies the scalar potential to be constant on each electrode, whereas (ii) implies
the constraint \eqref{weak_form_Ik_b}. It is worth noting that the $\phi_k$'s are unknown functions of time.



\begin{rem}
\label{remark:Ik_definition}
The linear functional:    
\begin{equation}
    \OP{I}_k: \mathbf{w} \in \hport \mapsto - \langle 1,\mathbf{w} \cdot \hat{\bf n} \rangle_{\Sigma_k} = - \int_{\Sigma_k} \mathbf{w} \cdot \hat{\bf n} \dd S \in \mathbb{R}
\end{equation}
gives the net electrical current entering $\Omega$ through the $k-$th electrode $\Sigma_k$ and due to $\bf w$.
\end{rem}

An effective approach to pose problem \eqref{weak_form_Ik_a}, \eqref{weak_form_Ik_b} consists in decomposing $\hport$ in terms of an orthogonal direct sum:
\begin{equation}
    \label{eq:direct_sum}
    \hport = H_0 \oplus H_D,
\end{equation}
where
\begin{equation}
H_0 = \left\{ \mathbf{u} \in \hport \, | \, \langle 1, \mathbf{u} \cdot \hat{\bf n} \rangle_{\Sigma_k} = 0, \ k=1,\ldots,N_e \right\},
\end{equation}
and $H_D=H_0^\bot$.

It results that $\dim \left( H_D \right)=N_e-1$, as from the following proposition.
\begin{prop}
\label{prop:dimension_HD}
\label{propDim}
    If $\Omega$ is a domain and on its boundary $\partial \Omega$ there are $N_e$ electrodes/ports $\Sigma_1, \ldots, \Sigma_{N_e}$, then $\dim \left( H_D \right)=N_e-1$.
\end{prop}
\begin{proof}
First, we notice that $H_0=\ker(\OP{I})$, where $\OP{I}$ is the operator giving the net electrical current entering the electrodes/ports $\Sigma_1, \ldots, \Sigma_{N_e}$
\begin{equation}
\label{eq:operator_electrodes}
        \OP{I}: \mathbf{w} \in \hport \mapsto
        \left[
        \begin{array}{c}
            \OP{I}_1 \mathbf{w} \\
            \vdots  \\
            \OP{I}_{N_e} \mathbf{w} \\
        \end{array}
        \right]
        \in \mathbb{R}^{N_e}.
\end{equation}

Then, we notice that $\dim \left( \im \left( \OP{I} \right) \right)=N_e-1$. Indeed, since $\Omega$ is assumed to be a domain (a non-empty connected open set), it results that $\sum_{k=1}^{N_e} \OP{I}_k \mathbf{w} = \int_{\partial \Omega} \mathbf{w} \cdot \hat{\bf n} \dd S =0$, via a divergence Theorem. This proves that the dimension of the preimage of $\im \left( \OP{I} \right)$ under $\OP{I}$ is equal to $N_e-1$.

The conclusions follows by recognizing that the latter, the preimage of $\im \left( \OP{I} \right)$ under $\OP{I}$, coincides with $H_D$.
\end{proof}

From \eqref{eq:direct_sum} and Proposition \ref{propDim} it follows that an arbitrary element of $\hport$ can be decomposed as:
\begin{equation}
    \mathbf{w}(x) = \sum_{k=1}^{N_e-1} w^D_k \mathbf{s}^D_k(x) + \mathbf{w}_0(x),
\end{equation}
where $\left\{ \mathbf{s}^D_1, \ldots, \mathbf{s}^D_{N-1} \right\}$ forms a basis of $H_D$ and $\mathbf{w}_0 \in H_0$. This decomposition induces a similar one on $L^2(0,T;\hport)$, that gives the unknown $\JJ$ to be represented as
\begin{equation}
\label{jdec}
    \mathbf{\JJ}(x,t) =  \mathbf{j}_D(x,t) +  \mathbf{j}_0(x,t) = \sum_{k=1}^{N_e-1} \jd^D_k(t) \mathbf{s}^D_k(x) + \mathbf{j}_0(x,t),
\end{equation}
being $\mathbf{j}_0 \in L^2(0,T;H_0)$.

Decomposition \eqref{jdec} allows to solve constraint \eqref{weak_form_Ik_b} easily. Indeed, when plugging \eqref{jdec} in \eqref{weak_form_Ik_b} one has
\begin{equation}
\label{impocurr}
    \Amat{F} \, \Avec{\jd^D}(t)=\Avec{i^D}(t)
\end{equation}
where $\JD^D(t)=\left[ \jd^D_1(t), \ldots, \jd^D_{N-1}(t) \right]^T$, $\mathbf{i}^D(t)=\left[ i_1^D(t), \ldots, i_N^D(t) \right]^T$ and the matrix $\mathbf{F} \in \mathbb{R}^{N_e \times (N_e-1)}$ is defined as:
\begin{equation}
    \label{eq:definition_electrode_incidence_matrix_continuum}
    F_{kl}=\OP{I}_k \mathbf{s}^D_l = \langle 1,  \mathbf{s}^D_l(x) \cdot \hat{\bf n} \rangle_{\Sigma_k}.
\end{equation}
Since $\mathbf{F}$ is a full rank matrix, the solution of \eqref{impocurr} is:
\begin{equation}
\label{alphaD}
    \Avec{\jd^D}(t)=\left( \Amat{F}^T \Amat{F} \right)^{-1} \Amat{F}^T \Avec{i^D}(t).
\end{equation}

At this stage, $\JJ_D$ is known and constraint \eqref{weak_form_Ik_b} is satisfied. The equation for $\JJ_0$ comes by imposing \eqref{weak_form_Ik_a} in $H_0$ rather than in $\hport$. By doing this, it turns out that
\begin{equation}
\label{weak_form_J0}
\left\langle \eta \JJ_0, \mathbf{w}\right\rangle + \langle\partial_{t}
\OA \JJ_0,\mathbf{w}\rangle  = - \left\langle \eta \JJ_D, \mathbf{w}\right\rangle - \langle\partial_{t}
\OA \JJ_D,\mathbf{w}\rangle -\langle\partial_{t} \AAA_S, \mathbf{w}\rangle \quad \forall\mathbf{w}\in H_0,
\end{equation}
where $\mathbf{j}_0 \in L^2(0,T;H_0)$.
Equation \eqref{weak_form_J0}, together with \eqref{alphaD}, is the governing equation for this contribution.

\begin{rem}
\label{remark:as_ad_a0}
    In the weak formulation \eqref{weak_form_Ik_a} or its equivalent \eqref{weak_form_J0}, it is convenient to split the source vector potentials in two terms: a first term $\AAA_D$ devoted to the \emph{direct} driving of electrical currents in the conducting domain $\Omega$ and another term $\AAA_0$ accounting for electrical currents that are not entering into $\Omega$. In particular the operator
    \begin{eqnarray}
    \OA_S:\mathbf{v} \in H_{L}(\Omega_S, \partial \Omega_s) \mapsto\frac{\mu_0}{4\pi}\int_{\Omega_S} \frac{\mathbf{v}(x^\prime) }{\vert\vert x-x^\prime\vert\vert}\ \textrm{d}V(x')\in L^2(\Omega_S;\R^3),
    \end{eqnarray}
    defines the magnetic vector potential in any point of space due to the source current distribution $\mathbf{v}(x) \in H_{L}(\Omega_S, \partial \Omega_S)$. It can be easily understood that by zeroing the $i_k$'s one get a vanishing $\AAA_D$, too.
\end{rem}

\subsection{Modal Decomposition}
The modal decomposition for eddy current problems has been introduced and developed in \cite{tamburrino2021themonotonicity} for problem \eqref{weak_form}.

The modal decomposition follows from space/time separation of variables applied to the source free problem. This is a key point to extend the results in \cite{tamburrino2021themonotonicity} of the modal decomposition for \eqref{weak_form} (absence of boundary electrodes) to the eddy current problem in the presence of injected electrical current via a set of boundary electrodes, as described in \eqref{weak_form_J0} and \eqref{alphaD}. Indeed, by zeroing \emph{all} source terms, problems \eqref{weak_form} and \eqref{weak_form_J0} become
\begin{align}
\label{source_free_a}
\left\langle \eta \JJ, \mathbf{w}\right\rangle + \langle\partial_{t}
\OA\JJ,\mathbf{w}\rangle & = 0, \quad \forall\mathbf{w}\in H_L(\Omega)\\
\label{source_free_b}
\left\langle \eta \JJ_0, \mathbf{w}\right\rangle + \langle\partial_{t}
\OA \JJ_0,\mathbf{w}\rangle  & = 0, \quad \forall\mathbf{w}\in H_0,
\end{align}
where $\JJ \in L^2(0,T;\hsol)$ and $\JJ_0 \in L^2(0,T;H_0)$. It is clear that the problems giving rise to the modal decomposition in the absence of boundary electrodes (see \eqref{source_free_a}) or in the presence of boundary electrodes (see \eqref{source_free_b}) are similar and, indeed, they differ only for the underlying functional space that is $L^2(0,T;\hsol)$ and $L^2(0,T;H_0)$, respectively. Luckily, the change of functional space has a minor impact and it is possible to prove that all the results developed for modal decomposition for problem \eqref{source_free_a} (see \cite{tamburrino2021themonotonicity}) can be extended to the problem \eqref{source_free_b}. Here, for the sake of completeness, we briefly summarize the key findings.

\medskip
\textbf{\emph{Modal Decomposition}}

The solution of problem \eqref{weak_form_J0} can be expanded in terms of a Fourier series as
\begin{eqnarray}
\label{sum_f}
\mathbf{j}_0\left(x,t\right)  =\sum_{n=1}^\infty i_n\left(t\right)\ {\bf j}_n (x)\quad\textrm{in}\ \Omega\times [0,T],
\end{eqnarray}
where $i_n \in L^2\left(0,T\right)$ and $\left\{ {\bf j}_n \right\}_{n=1}^{+ \infty}$ is the set of modes. In other terms, we have the following fundamental decomposition
\begin{equation*}
L^{2}\left(0,T; H_0 \right)=\overline{L^2\left(0,T\right)\otimes H_0}.    
\end{equation*}

\medskip
\textbf{\emph{A Generalized Eigenvalue Problem}}

The modes, and the related eigenvalues, are the solutions in $H_0$ of the following generalized eigenvalue problem:
\begin{equation}
\label{eq:generalized_eigenvalues}
\left\langle \OA \JJ,\mathbf{w}\right\rangle =\tau(\eta)\left\langle \eta\JJ, \mathbf{w}\right\rangle\quad \forall\ \mathbf{w}\in H_0.
\end{equation}

\medskip
\textbf{\emph{Properties of the modes and eigenvalues}}
\begin{enumerate}
\item the generalized eigenvalues and eigenvectors form countable sets:
$\left\{ \tau_{n}(\eta)\right\}_{n \in \mathbb{N}}$ and $\left\{ \mathbf{j}_{n}\right\}_{n\in \mathbb{N}}$;
\item the set of eigenvectors $\{\mathbf{j}_{n}\}_{n \in \mathbb{N}}$ forms a complete basis in $H_0$;
\item The elements ${\bf j}_n$ are orthogonal with respect to $\eta$, i.e. $\langle \mathbf{j}_n,\eta\mathbf{j}_{m}\rangle=0\quad \forall\ n\neq m$;
\item The elements ${\bf j}_n$ are orthogonal with respect to $\OA$, i.e. $\langle \mathbf{j}_n, \OA \mathbf{j}_m\rangle=0\quad \forall\ n\neq m$;
\item the eigenvalues can be ordered such that $\tau_{n}(\eta)\geq\tau_{n+1}(\eta)$;
\item $\tau_{n}(\eta)>0$ and $\lim_{n\rightarrow+\infty}\tau_{n}(\eta)=0$;
\item the following max-min variational characterization of $\tau
_{n}(\eta)$ holds:
\begin{eqnarray}
\label{second_charac}
\tau_{n}(\eta) = \max_{ \dim \left( U \right) = n } \min_{ \mathbf{j} \in U} \frac{\langle \OA  \mathbf{j}, \mathbf{j}\rangle  }{\langle \eta  \mathbf{j}, \mathbf{j}\rangle },
\end{eqnarray}
where $U$ is a linear subspace of $H_0$;
\item the eigenvalues satisfies a Monotonicity Principle w.r.t. the electrical resistivity $\eta$:
\[
\eta_{1} \leq\eta_{2}
\ \textrm{a.e. in}\ \Omega\quad \Longrightarrow\quad\tau_{n}\left(  \eta_{1}\right)  \geq\tau
_{n}\left(  \eta_{2}\right)  \ \forall n\in%
\mathbb{N}
,
\]
where $\eta_{1},\eta_{2}\in L^{\infty}_+ \left( \Omega \right)$, $\tau_{n}\left(  \eta_{1}\right)$ and $\tau_{n}\left(  \eta_{2}\right)  $ are the $n-$th eigenvalues related to $\eta_{1}$ and $\eta_{2}$, respectively;
\end{enumerate}

\medskip
\textbf{\emph{Decoupling the Eddy Current Problem}}

Once the modes $\left\{ \mathbf{j}_{n}\right\}_{n\in \mathbb{N}}$ have been computed for a given electrical resistivity $\eta$, the eddy current density is determined from the $i_n$s. The evolution of these functions of the time is determined by the following \emph{first order linear} ODE:
\begin{eqnarray}
\label{ode_f2}
l_n i_n' + r_n i_n=\mathcal{E}^D_n + \mathcal{E}^S_n \quad\forall n \in \mathbb{N}.
\end{eqnarray}
In (\ref{ode_f2})
\begin{align}
    r_n & = \left\langle \eta {\bf j}_n, {\bf j}_n\right\rangle, \
    l_n  =\left\langle \OA {\bf j}_n,{\bf j}_n \right\rangle
    \\
    \mathcal{E}^D_n & =- \left\langle \eta \JJ_D, {\bf j}_n \right\rangle - \langle\partial_{t}
\OA \JJ_D,{\bf j}_n \rangle
    \\
    \mathcal{E}^S_n & =-\langle\partial_{t} \AAA_S, {\bf j}_n \rangle.
\end{align}

{Hereafter we assume the unit of $i_n$ to be that of an electrical current (A), in the International System of Units (SI). The unit of $\textbf{j}_n$ is, therefore, the inverse of an area ($\textrm{m}^{-2}$), the unit of $r_n$ is electrical resistance ($\Omega$), the unit of $l_n$ is henry (H), and the unit of $\mathcal{E}^D_n$ and $\mathcal{E}^S_n$ is volt (V).}

Equation \eqref{ode_f2} is the key to fast solution of eddy current problems: once the modes $\left\{ \mathbf{j}_{n}\right\}_{n\in \mathbb{N}}$ are available, the computation of the $i_n$s is modeled by decoupled ODEs. Their solution can be computed (i) separately, (ii) with different integration methods, (iii) in a parallel environment, and (iv) in a selective manner.

\section{Modal Decomposition in the discrete model}
\label{sec:discrete_model}
\subsection{Discrete model}
The eigenvectors $\mathbf{j_n}(x)$ introduced in (\ref{sum_f}) are in general not known analytically except that for simple geometries, \textit{e.g.} the case of the sphere is commented in \cite{Pascale2020}. In general, for conductors of arbitrary shape or non-uniform conductivity we will have to find suitable approximations of the original problem. In particular, for each functional space $V$ defined in equations (\ref{eq:L_infinity})-(\ref{eq:HL}) we may consider a family of finite dimensional subspaces $V_h$ depending on a positive parameter $h>0$. In the limit $h \rightarrow 0$ the dimension of the subspace $V_h$ tends to infinity and the vector subspace $V_h$ tends to the original functional space $V$. Typically $h$ is representative of the mesh size used to discretize the domain $\Omega$. In the following we fix the value of the $h$ parameter and omit it from notation. In Reference \cite{Albanese1988} you can find the construction of a finite dimensional approximation of the functional space $H_L(\Omega,\Sigma)$, based on the use of edge shape functions. For what we aim to discuss in this Section, it is sufficient to assume that we are able to identify a set of $N$ linearly independent vectors $\{ \mathbf{w_k}(x) \in H_L(\Omega,\Sigma) \}$ which generates a finite dimensional approximation $X$ of $H_L(\Omega,\Sigma)$:
\begin{equation}
        \label{eq:finite_dimensional_vector_space}
        X = \mathrm{span}\left( \mathbf{w_k}(x) , \, k=1,\cdots,N\right) \subset H_{L}(\Omega,\Sigma) \, .
\end{equation}
Any vector $\Avec[]{v} \in X$ can be clearly identified with its components in the vector basis $\{ \mathbf{w_k} \}$, \textit{i.e.} $\mathbf{v}(x) \equiv (I_1, \cdots, I_N) \in \mathbb{R}^{N}$. The construction of the linearly independent vectors $\mathbf{w_k}$ which span $X$ is carried out in such a way that 
\begin{equation}
    H_D \subset X 
\end{equation}
\textit{i.e.} the finite-dimensional space $H_D$ of current distributions crossing the $\Sigma_k$ electrodes is certainly contained in our finite dimensional approximation $X \subseteq H_L(\Omega,\Sigma)$. 

In the discrete context, we look for the $\mathbf{j}(x,t) \in L_2(0,T,X)$ satisfying equations (\ref{weak_form_Ik_a})-(\ref{weak_form_Ik_b}) for any test vector in the basis of $X$, \textit{i.e.} for any vector $\EulerVec[]{w}(x) \in \{ \EulerVec[]{w_k}(x) , k=1,\cdots, N\}$. Explicitly the problem takes the form
\begin{align}
    \label{eq:discrete_version_1a}
    \Amat[]{R} \Avec[]{I} + \Amat[]{L} \deriv[]{}{t} \Avec[]{I} + \Amat[]{E}^T \Avec[]{\phi} = \Avec[]{V_0} + \Avec[]{V_D} \\
    \label{eq:discrete_version_1b}
    \Amat[]{E} \Avec[]{I} = \Avec[]{i_D}
\end{align}
Here we introduced the resistance $\Amat{R}$, inductance $\Amat{L}$ and the electrodes-currents \textit{reduced} incidence matrix $\Amat{E}$, besides the vectors of applied voltages $\Avec{V_0}$ and $\Avec{V_D}$. Given the basis of $X$, $\{ \mathbf{w_k} , k =1,\cdots,N \}$ , the coefficients of the resistance and inductance matrix are defined by
\begin{align}
    R_{jk}   = \braket[]{\eta \EulerVec{w}_j}{\EulerVec{w}_k}    \, , j,\,k=1,\cdots,N, \\
    L_{jk}   = \braket[]{ \OA \EulerVec{w}_j}{\EulerVec{w}_k}    \, , j,\,k=1,\cdots,N.
\end{align}
By definition the above matrices are symmetric and positive definite. In particular the matrix $\Amat[]{R}$ is positive definite if and only if the resistivity $\eta$ is strictly positive, circumstance which we assume valid here. The reduced incidence matrix $E$, as clear by inspection of equation (\ref{weak_form_Ik_a}) and remark \ref{remark:Ik_definition}, is defined as
\begin{equation}
\label{eq:definition_E}
    \left(E  \right)_{j,k} = \mathcal{I}_j \mathbf{w}_k = - \braket[]{1}{\mathbf{w_k}\cdot \mathbf{\hat{n}}}_{\Sigma_j} \quad \forall j = 1, \cdots, N_e-1 , \quad \forall k = 1, \cdots, N\, .
\end{equation}
{Notice that we arbitrarily excluded one electrode from the incidence matrix. The additional row would have indeed resulted to be linearly dependent from the others, as demonstrated in proposition \ref{prop:dimension_HD}. The currents impressed in the electrodes $\Avec[]{i_D}$ were already introduced in the continuous model. Further, we defined the induced voltages considering separately current sources distributions which are entering $\Omega$ and the current sources which are solenoidal in $\Omega_s$, as indicated in remark \ref{remark:as_ad_a0}:
\begin{align}
    \label{eq:Vd_voltage}
    (V_D)_j = \braket[]{- \partial \EulerVec[]{a_d} / \partial t }{\EulerVec{w_j}} \\
    \label{eq:V0_voltage}
    (V_0)_j = \braket[]{- \partial \EulerVec[]{a_0} / \partial t }{\EulerVec{w_j}}
\end{align}

Exactly as in the continuous model, it is convenient to introduce an orthogonal decomposition to distinguish between currents internal to the conducting domain and currents associated to fluxes through electrodes. In different words, we are looking for the discrete analog of decomposition (\ref{eq:direct_sum}),
\begin{equation}
    \label{eq:direct_sum_discrete}
    X = X_0 \oplus X_D \, ,
\end{equation}
In equation (\ref{eq:direct_sum_discrete}) $X_0$ is the discrete approximation of $H_0$, \textit{i.e.}
\begin{equation}
    X_0 = \{ \mathbf{v}(x) \in X : \mathbf{v}(x) \cdot \hat{\mathbf{n}} = 0 \quad \forall x \in \partial \Omega \} \, .
\end{equation}
Notice that the reduced incidence matrix $\Amat[]{E}$ is the discrete version of the operator $\mathcal{I}$ introduced in equation (\ref{eq:operator_electrodes}), although excluding the current flux through an arbitrary electrode. In particular in the kernel of $\Amat[]{E}$ we will have all those vectors of $X$ which do not determine current fluxes through the electrodes, \textit{i.e.}
\begin{equation}
    X_0 \equiv \ker{(\Amat[]{E})} \, .
\end{equation}
Clearly $X_0$ has dimension $N_0 = N - (N_e-1)$. We collect the basis vectors for the null space of $\Amat[]{E}$ as columns of the matrix $\Amat[]{K} \in \mathbb{R}^N \times \mathbb{R}^{N_0}$. By simple linear algebra arguments, the orthogonal subspace $X_D$ is generated by the image of $\Amat[]{E}^T$:
\begin{equation}
    X_D \equiv \im{(\Amat[]{E}^T)} \, .
\end{equation}
which has dimension $N_e-1$. Since $X_D$ has the same dimensions as $H_D$, and linear combinations of vectors in $X_D$ are sufficient to describe any admissible current flow across the electrodes $\Sigma_k$, it is clear that $X_D \equiv H_D$. The columns of $\Amat[]{E}^T$ provide then a suitable set of basis vectors for $H_D$, as represented in the basis $\{ \mathbf{w_k}, k=1,\cdots, N \}$ of $X$. The equivalent for the continuous representation (\ref{jdec}) is here given by
\begin{equation}
    \Avec[]{I} = \Amat[]{K} \Avec[]{I_i} + \Amat[]{E}^T \Avec[]{I_e}
\end{equation}
The vector $\Avec[]{I_i} \in L_2(0,T,\mathbb{R}^{N_0})$ represent the internal currents not crossing any electrode, while the vector $\Avec[]{I_e} \in L_2(0,T,\mathbb{R}^{N_e-1})$ accounts for the current distributions impressed via the electrodes. Equation (\ref{eq:discrete_version_1b}) immediately reveals, since $\Amat[]{E}$ is a full-rank matrix:
\begin{equation}
    \label{eq:current_electrodes}
    \Avec[]{I_e} = ( \Amat[]{E} \, \Amat[]{E}^T )^{-1} \Avec[]{i_D} \, \rightarrow  \, \Avec[]{I} = \Amat[]{K} \Avec[]{I_i} + \Amat[]{E}^T ( \Amat[]{E} \, \Amat[]{E}^T )^{-1} \Avec[]{i_D} \, .
\end{equation}
{The vector $\Avec{I}_e$ describes the injected current distribution within structures in the basis $\{ \mathbf{v}_j(x) = \sum_k E_{j,k} \mathbf{w}_k(x), \, j=1,\cdots,N_e-1 \}$, in the same stream as the vector $\Avec[]{\alpha}^D$ in equation (\ref{alphaD}) describes the injected current distribution in the vector basis $\{ \mathbf{s}_k^D(x), \, k=1,\cdots,N_e-1 \}$. These current distributions can be different, anyway they reproduce exactly the same current flux through boundary electrodes.} Projecting equation (\ref{eq:discrete_version_1a}) to the subspace $X_0$, which corresponds to left multiplication by $\Amat[]{K}^T$, we find
\begin{equation}
    \label{eq:discrete_version_2}
    \Amat[]{R_i} \Avec[]{I_i} + \Amat[]{L_i} \deriv[]{}{t} \Avec[]{I_i} = - \Amat[]{R_{ie}} \Avec[]{I_e} - \Amat[]{L_{ie}} \deriv[]{}{t} \Avec[]{I_e} + \Avec{V_{D,i}} + \Avec{V_{0,i}} \, , 
\end{equation}
where
\begin{align}
    \Amat[]{R_i} = \Amat[]{K}^T \Amat[]{R} \Amat[]{K} \\
    \Amat[]{L_i} = \Amat[]{K}^T \Amat[]{L} \Amat[]{K} \\
    \Amat[]{R_{ie}} = \Amat[]{K}^T \Amat[]{R} \Amat[]{E}^T \\
    \Amat[]{L_{ie}} = \Amat[]{K}^T \Amat[]{L} \Amat[]{E}^T \\
    \Avec{V_{D,i}} = \Amat[]{K}^T \Avec{V_{D}} \\
    \Avec{V_{0,i}} = \Amat[]{K}^T \Avec{V_{0}}    
\end{align}
The problem of finding $\Avec[]{I} \in L_2(0,T, \mathbb{R}^N)$ satisfying equations (\ref{eq:discrete_version_1a})-(\ref{eq:discrete_version_1b}) has been reduced to the problem of finding $\Avec[]{I_i} \in L_2(0,T, \mathbb{R}^{N_0})$ satisfying equation (\ref{eq:discrete_version_2}) following the same ideas and the same strategy which lead us to reduce the problem of finding $\mathbf{j}(x,t) \in H_L(\Omega, \Sigma)$ satisfying equations (\ref{weak_form_Ik_a})-(\ref{weak_form_Ik_b}) to the problem of finding $\mathbf{j_0}(x,t) \in H_0$ satisfying equation (\ref{weak_form_J0}) in the continuous framework. 

The dynamical properties of the system are completely described by the resistance matrix $\Amat[]{R_i}$ and the inductance matrix $\Amat[]{L_i}$, which are both symmetric and positive definite. The matrices $\Amat[]{L_{ie}}$ and $\Amat[]{R_{ie}}$ account for the inductive and resistive coupling of the electrode-to-electrode currents $\Avec[]{I_e}$ with the internal currents $\Avec[]{I_i}$. 
\begin{rem}
    In definitions (\ref{eq:Vd_voltage})-(\ref{eq:V0_voltage}) we explicitly considered the magnetic vector potential due to external current sources as sum of two contributions, as suggested in Remark \ref{remark:as_ad_a0}. In practical applications the vectors of induced voltages $\Avec[]{V^D}(t)$ and $\Avec[]{V^0}(t)$ may be obtained from information on the current density within $\Omega_S$. In general it is possible to identify a finite subset of current distributions defining accurately the overall current density in the source domain, \textit{i.e.}
    \begin{equation}
    \label{eq:jsource}
        \mathbf{j_s}(x,t) = \sum_{k=1}^{N_e-1}{ i^D_k(t) \mathbf{u}^D_k(x)} + \sum_{k=1}^{N_0}{ i^0_k(t) \mathbf{u}^0_k(x)}
    \end{equation}
    where the vector fields $\mathbf{u^0}_k(x) \in H_L(\Omega_S)$ and $\mathbf{u^D}_k(x) \in H_L(\Omega_S, \Sigma)$ are prescribed forcing terms. In particular the $N_e-1$ current distributions $\mathbf{u^D}_k$ are associated to the net current fluxes through the interface electrodes with the conducting domain $\Sigma_k$. Representation (\ref{eq:jsource}) in this discrete context induce to define the mutual inductance matrices: 
    \begin{eqnarray}
      M^D_{ij} = \braket[]{\OA_S \mathbf{s_k}}{\mathbf{w_j}} \, , \\
      M^0_{ij} = \braket[]{\OA_S \mathbf{u_k}}{\mathbf{w_j}} \, .
    \end{eqnarray}
    where $\mathbf{s_k}$ are current distributions in $\Omega_S$ which drain current in the conducting domain via the electrodes $\Sigma_k$ while the $\mathbf{u_k}$ are solenoidal current distributions in $\Omega_S$. These mutual inductance matrices can be used to express the induced voltages in passive structures in terms of time variations of the current sources,
    \begin{eqnarray}
    \label{eq:induced_voltages_d}
      \Avec{V^D} = - \Amat[]{M^D} \deriv[]{}{t} \Avec[]{i^D} \, , \\
    \label{eq:induced_voltages_0}
      \Avec{V^0} = - \Amat[]{M^0} \deriv[]{}{t} \Avec[]{i^0} \, . 
    \end{eqnarray}    
\end{rem}
\begin{rem}

It is possible to plug equations (\ref{eq:current_electrodes}) and (\ref{eq:induced_voltages_d})-(\ref{eq:induced_voltages_0}) directly into the circuit equation (\ref{eq:discrete_version_2}) to make the current sources appear explicitly,
    \begin{eqnarray}
    \label{eq:circuit_equations_v2}
    \Amat[]{R_i} \Avec[]{I_i} + \Amat{L_i} \deriv[]{}{t} \Avec[]{I_i} = - \Amat[]{R_{D}} \Avec[]{i_D} - \Amat[]{L_D} \deriv[]{}{t} \Avec[]{i_D} - \Amat[]{M^0} \deriv[]{}{t} \Avec{i_0} \, ,    
    \end{eqnarray}
where we defined
    \begin{eqnarray}
        \Amat[]{R_D} = \Amat[]{R_{ie}} \, \left( \Amat[]{E} \Amat[]{E}^T \right)^{-1}  \, , \\
        \Amat[]{L_D} = \Amat[]{L_{ie}} \, \left( \Amat[]{E} \Amat[]{E}^T \right)^{-1} + \Amat[]{M^D} \, .
    \end{eqnarray}
\end{rem}

\subsection{Modal decomposition}

In this subsection we introduce the analog of the generalized eigenvalue problem (\ref{eq:generalized_eigenvalues}) within the \textit{discrete} context developed in this Section. In the next Section we will illustrate how to use generalized eigenvalues and eigenvectors to accelerate standard time-domain eddy currents simulations, taking an example from real applications. Zeroing all the source terms in equation (\ref{eq:circuit_equations_v2}), we are left with the homogeneous problem
\begin{equation}
    \label{eq:source_free_discrete}
    \Amat[]{L_i} \deriv[]{}{t} \Avec[]{I_i} + \Amat[]{R_i} \Avec[]{I_i} = \Avec[]{0} \, .
\end{equation}
The problem of finding $\Avec[]{I_i}(t) \in L_2(0,T;\mathbb{R}^{N_0})$ which satisfies (\ref{eq:source_free_discrete}) is the discrete analog of the problem of finding $\mathbf{j_0} \in L_2(0,T;H_0)$ which satisfies (\ref{source_free_b}), under corresponding initial conditions. Thanks to the Galerkin formulation adopted, the inductance matrix $\Amat{L}$ inherits the properties of symmetry and positive-definiteness from the operator
\begin{equation}
    (\mathbf{w_j}, \mathbf{w_k}) \in H_L(\Omega,\Sigma) \times H_L(\Omega,\Sigma) \mapsto \braket[]{\OA \mathbf{w_j}}{\mathbf{w_k}} \in \mathbb{R} \, .
\end{equation}
These properties are further preserved when the action of the operator is restricted to vector fields in $H_0$. Hence, the matrix $\Amat[]{L_i}$, which restricts the action of the matrix $\Amat[]{L}$ to vectors of $X_0$, is still symmetric and positive-definite. 
Similarly, the resistance matrix $\Amat{R}$ inherits the property of symmetry from the operator
\begin{equation}
    (\mathbf{w_j}, \mathbf{w_k}) \in H_0 \times H_0 \mapsto \braket[]{\eta \mathbf{w_j}}{\mathbf{w_k}} \in \mathbb{R} \, .    
\end{equation}
The resistance matrix $R$ is symmetric and positive definite in case of strictly passive conductors, or at least positive semi-definite in case some ideal conductors are present. Again $\Amat[]{R_i}$ inherits the properties of symmetry and positive definiteness from $\Amat[]{R}$, being its projection to the space of internal currents.

Since both $\Amat[]{L_i}$ and $\Amat[]{R_i}$ are real symmetric matrices and at least $\Amat[]{L_i}$ is positive definite, the generalized eigenvalue problem 
\begin{equation}
    \label{eq:generalized_eigenvalues_discrete}
    \Amat[]{L}_i \Avec[]{V_\tau} = \tau \Amat[]{R}_i \Avec[]{V_\tau}
\end{equation}
admits only real non-negative eigenvalues $\tau$, which are strictly positive in case $R$ is positive definite. Moreover the generalized eigenvectors $\Avec[]{V_\tau}$ constitute a basis of $\mathbb{R}^{N_0}$ whose elements are $L$- and $R$-orthogonal. The Monotonicity principle still holds in the discrete finite context. In particular, re-scaling the resistance matrix $\Amat[]{R}$ by some multiplicative factor $k$ the eigenvectors are preserved and the eigenvalues are re-scaled:
\begin{equation}
    \tau_n(k) = \frac{1}{k} \tau_n(1) \, .
\end{equation}
We collect the eigenvectors $V_\tau$ into the matrix $\Amat[]{V}$. We denote with a tilde vectors and matrices after we change the basis of $\mathbb{R}^{N_0}$ from the canonical one to the eigenvector basis. The original coupled system of first order Ordinary Differential Equations (\ref{eq:circuit_equations_v2}) is decoupled into $N_0$ non-interacting ODEs when moving to the basis of generalized eigenvectors defined in (\ref{eq:generalized_eigenvalues_discrete}):
\begin{equation}
    \label{eq:modal_decomposition_discrete}
    L_n I_n' + R_n I_n = E_n \, , \forall n=1,...,N_0 \, .
\end{equation}
In equation (\ref{eq:modal_decomposition_discrete}) $L_n = \Avec[]{V_n}^T \, \Amat[]{L} \, \Avec[]{V_n}$, $R_n = \Avec[]{V_n}^T \, \Amat[]{R} \, \Avec[]{V_n}$ and the forcing term is $E_n = - \Avec[]{V_n}^T ( \Amat[]{L^D} \Avec[]{i_D}' + \Amat[]{R^D} \Avec[]{i_D} + \Amat{M_0} \Avec{i_0}') $. The time constant associated to the $n$-th eigenmode is clearly given by $\tau_n = L_n / R_n$. We can advance separately in time each eigenmode of the dynamical system, providing the engineer with several possible exact or approximate methods for the fast time integration of the system. 

\subsection{Time integration scheme}

The modal decomposition of the dynamical system opens a new range of opportunities for the time integration of the dynamical system, since analytical solutions to polynomial input signals become now easily manageable. We postpone the exploration of analytical and semi-analytical techniques to future work on the subject, while here we explore the computational advantages obtained via modal decomposition even in the framework of the same time integration technique. 

\subsubsection{Cholesky decomposition}
The standard approach to find the solution $\Avec[]{I_i}(t) \in L_2(0,T,\mathbb{R}^{N_0})$ of the discrete eddy current problem (\ref{eq:circuit_equations_v2}) is to use a $\theta$-method, so to allow eventually for larger time steps:
\begin{equation}
\begin{aligned}
    \left( \Amat[]{{L_i}}  + \Delta t \, \theta \, \Amat[]{{R_i}} \right) \,  \delta \Avec[]{I} = & - \Delta t  \Amat[]{R_i} \, \Avec[]{I_i}^{(k)} \\
    & - \Delta t \Amat[]{{R}_{D}} \, \Avec[]{i_D}^{(k)} -  \left( \Amat[]{L_{D}}  + \Delta t \, \theta \, \Amat[]{R_{D}} \right) \, \delta \Avec[]{{i}_D} - \Amat[]{M_0} \delta\Avec[]{i_0}     
\end{aligned}
\label{eq:time_integration}
\end{equation}
where
\begin{equation}
    \delta \Avec[]{x} = \left( \Avec[]{x}^{(k+1)} - \Avec{x}^{(k)} \right) \, .
\end{equation}
Here $\theta$ is a numerical parameter to choose in the interval $[0,1]$, the implicit Euler scheme corresponding to the case $\theta = 1$. Since the matrix $\Amat[]{Z} = \left( \Amat[]{L} + \Delta t \, \theta \Amat[]{R} \right)$ is symmetric and positive definite, a convenient method to solve (\ref{eq:time_integration}) is based on the Cholesky decomposition. For fixed time integration scheme (\textit{i.e.} fixed $\theta$ parameter), and for fixed time step $\Delta t$, we can compute the Cholesky decomposition of $\Amat[]{Z}$ once for all, \textit{i.e.} $\Amat[]{Z} = \Amat[]{C} \, \Amat[]{C}^T$, where $\Amat[]{C}$ is a real lower triangular matrix. At each time step we just need to solve the lower triangular problem $\Amat[]{T} \, \Avec[]{y} = \Avec[]{b}$ (forward substitution) and the upper triangular problem $\Amat[]{T}^T \, \delta \Avec[]{I} = \Avec[]{y}$ (backward substitution). We defined with the symbol $\Avec[]{b}$ the right hand side of the linear algebraic system, depending solely to the currents at the previous time step $\Avec[]{I_i}^{(n)}$ and to the forcing terms. The overall computational cost of a simulation will be proportional to the cost of computing the Cholesky decomposition for the dynamic matrix $\Amat[]{Z}$, plus the cost of the backward and forward substitutions at each time step. 

\subsubsection{Modal decomposition}

Applying the $\theta$-method described in previous section to the set of $N_0$ linearly independent equations (\ref{eq:modal_decomposition_discrete}), we find
\begin{equation}
    \left( L_n + \Delta t \, \theta \, R_n  \right) \delta I_n = - \Delta t R_n I_n^{(k)} + \theta E_n^{(k+1)} + (1-\theta) E_n^{(k)}
\end{equation}
which is immediate to invert. The overall computational cost of a simulation will be proportional to the cost of computing the generalized eigenvalues and eigenvectors, performing necessary projections (see definition of $E_n$ below equation (\ref{eq:modal_decomposition_discrete})), plus the cost of solving $N$ independent linear scalar equations at each time step.

\section{Results}
\label{sec:results}

In this last section, the performances of the proposed method are evaluated on a practical problem in the framework of Non destructive Testing (NdT) based on eddy currents. Eddy currents can be effectively exploited to detect and retrieve useful information about defects in metallic samples. Indeed, the presence of flaws produces a perturbation in the eddy current circulation, resulting in a peculiar reaction magnetic flux density. In this context, numerical simulations play a paramount role, being an essential tool in designing and optimizing the equipment required in the overall inspection process.

However, the computational time required by such kind of numerical simulations can be very relevant and, hence, there is a strong demand of new numerical techniques able to accelerate the process. In Section~\ref{sec:comp}, it is shown that the proposed method is able to drastically reduce the computational time with respect to the already available methods.

\subsection{Proposed NdT method}
{The sample under test is a commercial cylindrical tube with an outer diameter of $168.28\,\text{mm}$, a thickness of $3\,\text{mm}$ and a length of $1.5\,\text{m}$. The excitation is given by alternating currents injected through two diametrically opposite pairs of conductive legs at the ends of the tube (see Figure~\ref{fig:tube}). The electrical resistivity is $\eta=1.09\times 10^{-6}\,\Omega\cdot m$, which is the one of a particular aluminium alloy (Inconel 600), a material widely used in manufacturing steam generators tubes for nuclear power plants~\cite{LIM200397}.}

\begin{figure}[htp]
    \centering
    \subfloat[][\emph{Sample geometry}]
    {\includegraphics[width=0.21\linewidth]{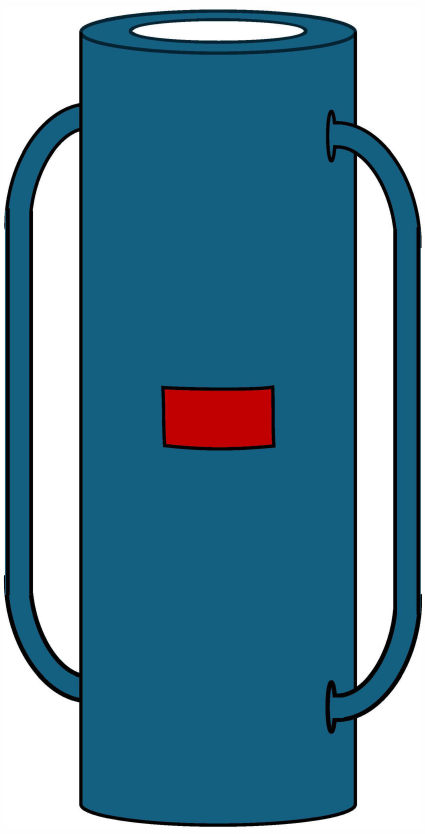}} \quad\quad\quad
    \subfloat[][\emph{Excitation block scheme}]
    {\includegraphics[width=0.45\linewidth]{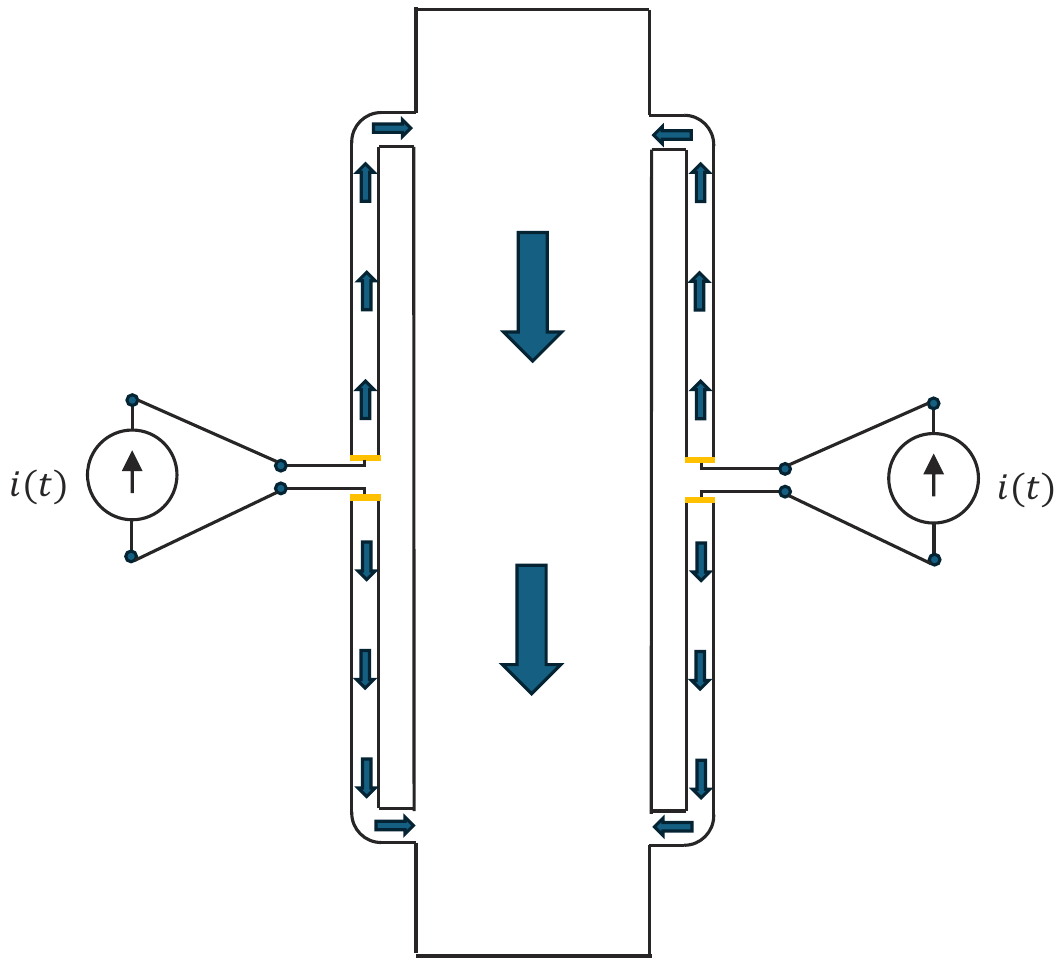}}
    \caption{Left: Geometry of the sample, in red the position of the analyzed cracks. Right: Injection of the exciting current. In yellow the two pairs of electrodes, the blue arrow represents the current density (not in scale).}
    \label{fig:tube}
\end{figure}

{In this work, the aim is to retrieve information about flaws by means of magnetic flux density measurements. Specifically, the response of a defect-free sample is simulated and used as reference. Subsequently, the responses from configurations with different defects are compared with the reference one, analyzing the behaviour of the system at several different excitation frequencies. For a detailed description of the data processing see Section~\ref{sec:postpro}.}

\subsection{Defects definition}
{Three different configurations are analyzed: a defect-free configuration used as reference (background signal), a tube with a surface crack situated in the outermost position and a configuration with a deep crack located in the innermost tube position. Figure~\ref{fig:cracks} shows the details of the two cracks.}

\begin{figure}[htp]
    \centering
    \subfloat[][\emph{Deep Crack: lateral view}]
    {\includegraphics[width=.45\textwidth]{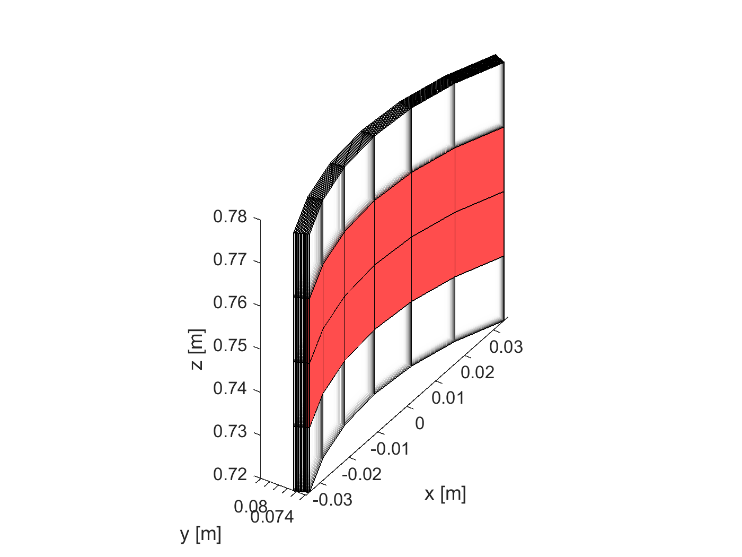}} \quad
    \subfloat[][\emph{Deep Crack: top view}]
    {\includegraphics[width=.20\textwidth]{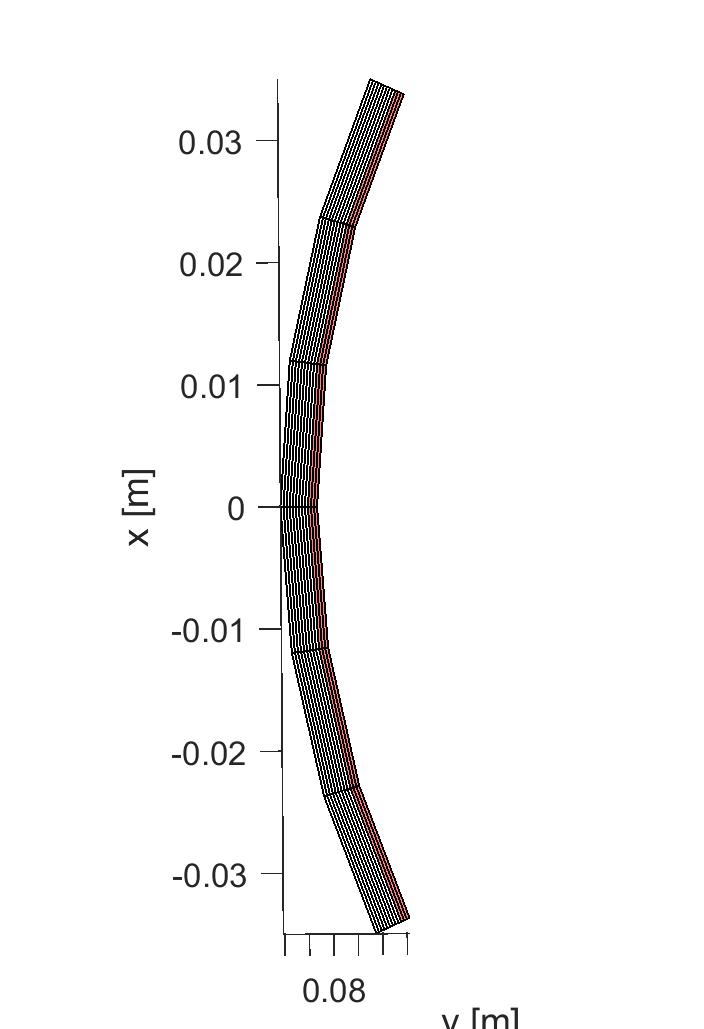}} \\
    \subfloat[][\emph{Surface Crack: lateral view}]
    {\includegraphics[width=.30\textwidth]{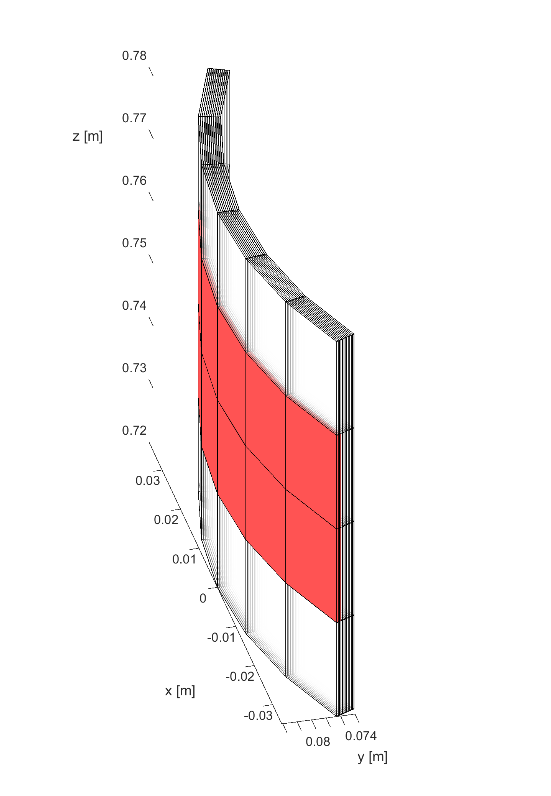}} \quad
    \subfloat[][\emph{Surface Crack: top view}]
    {\includegraphics[width=.20\textwidth]{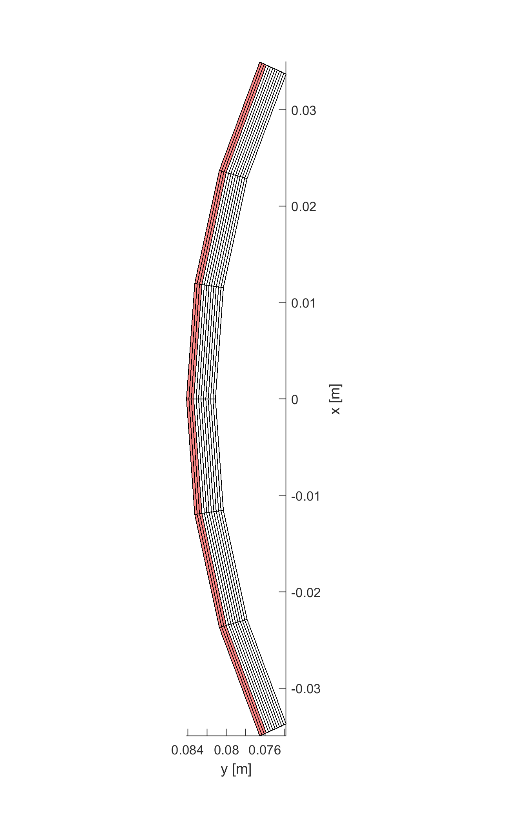}}
    \caption{Simulated cracks}
    \label{fig:cracks}
\end{figure}

{The different location of the cracks is dictated by the need to verify the performances of the eddy-currents testing system on all the range of frequencies selected. Indeed, by a physical point of view, the inner cracks react only to lower frequencies, due to the penetration depth of the magnetic field, while surface cracks can be revealed by both lower and higher frequencies, but with different kinds of magnetic response.}

\subsection{Excitation current}
{In order to retrieve data from the cracks at different frequencies, a multi-sine waveform is adopted as excitation current}
\begin{equation}\label{eqn:exc}
    i(t)=\sum_{k=1}^{N_s}I_k\sin\left(2\pi f_k t+\phi_k\right).
\end{equation}
{Such signals are commonly used in eddy current testing, since they can significantly reduce measurement time in practical applications, giving the opportunity to collect information from various frequencies in only one measurement step~\cite{Sardellitti_2022}. On the other hand, this excitation poses serious challenges from a computational point of view. Indeed, a proper time-domain simulation of the sample response requires long simulation times (due to the tones at lower frequencies) and small time steps (due to the tones at higher frequencies).}

{For the case of interest, $N_s=30$ tones with equal amplitude $I_k=4\,\text{A}$ are adopted. The tone frequencies are $F=\{$1, 50, 100, 150, 250, 300, 400, 600, 700, 800, 900, 1000, 1500, 2000, 2500, 3000, 3500, 4000, 4500, 5000, 5500, 6000, 6500, 7000, 7500, 8000, 8500, 9000, 9500, 10000$\}$ Hz}
{and they are selected such that the penetration depth varies in a range compatible with the thickness of the tube. Regarding the phases $\phi_k$, they are chosen in order to avoid a large peak to RMS ratio of $i(t)$. As described in~\cite{Sardellitti_2022}, an optimal choice is}
\begin{equation*}
    \phi_k=-\pi\frac{k(k-1)}{N_s},
\end{equation*}
{which ensures that $i(t)$ has a quite constant envelope, a fundamental requirement to optimize the range of the analog to digital converter in acquisition. The obtained excitation signal is depicted in Figure~\ref{fig:currentsignal}.}

\begin{figure}
    \centering
    \includegraphics[width=0.75\textwidth]{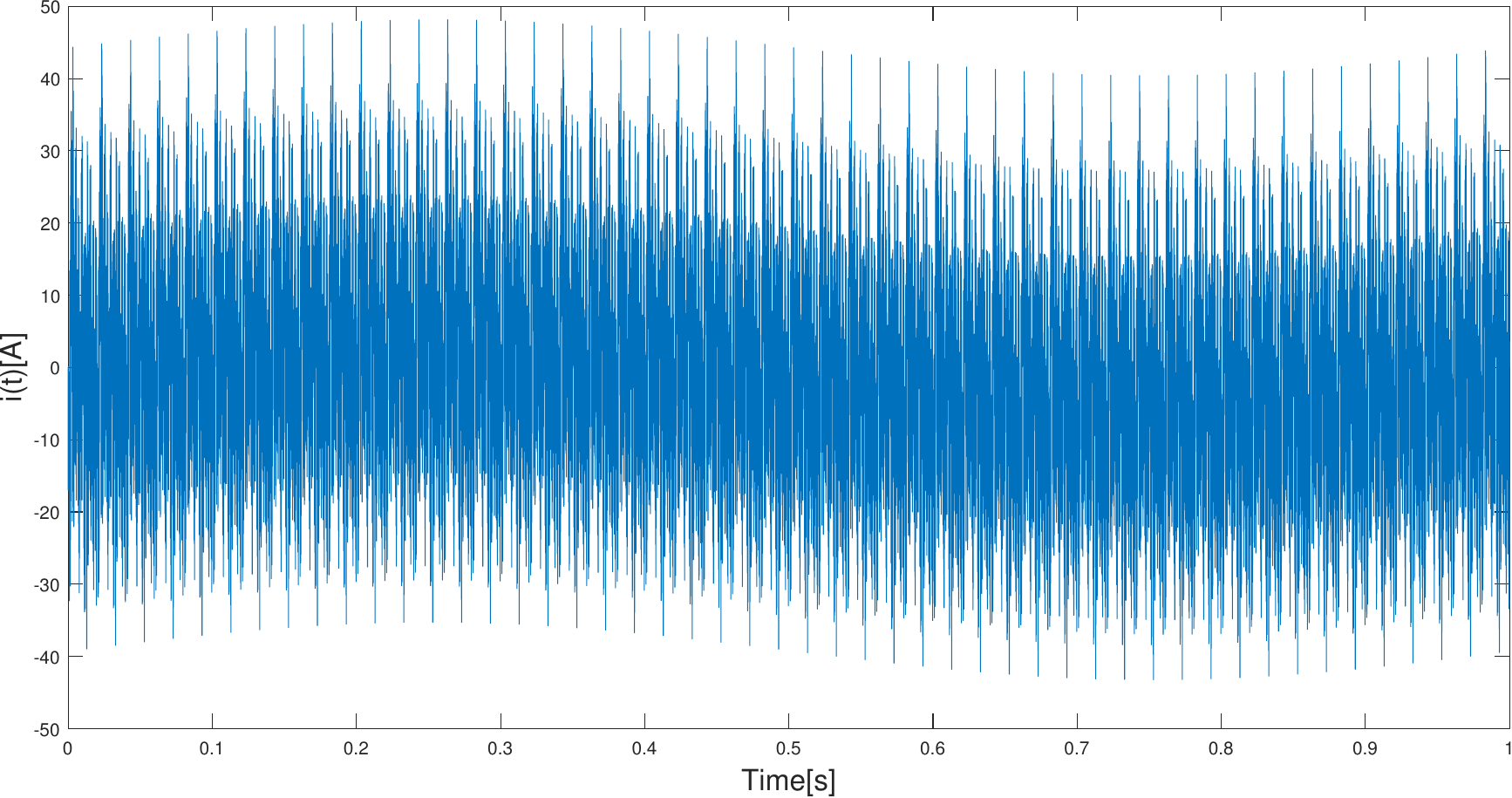}
    \caption{Injected current. {The current has a quite constant envelope, with a maximum value around $48\,\text{A}$ and a minimum value around $-43\,\text{A}$.}}
    \label{fig:currentsignal}
\end{figure}

\subsection{Sensing points definition}
{In order to characterize the response of the two cracks analyzed, the values of the magnetic flux density are collected from $N_p=301$ points placed around the tube, simulating the measurements made by an array of magnetic sensors. Specifically, the points are distributed at a fixed distance of $3\,\text{mm}$ from the tube, a fixed height of $750\,\text{mm}$ from the bottom of the tube and they span for $\pi$ radians around the defect (see Figure~\ref{fig:scan}).}

\begin{figure}[htb]
    \centering
    \subfloat[][\emph{Sensing points}]
    {\includegraphics[width=0.45\linewidth]{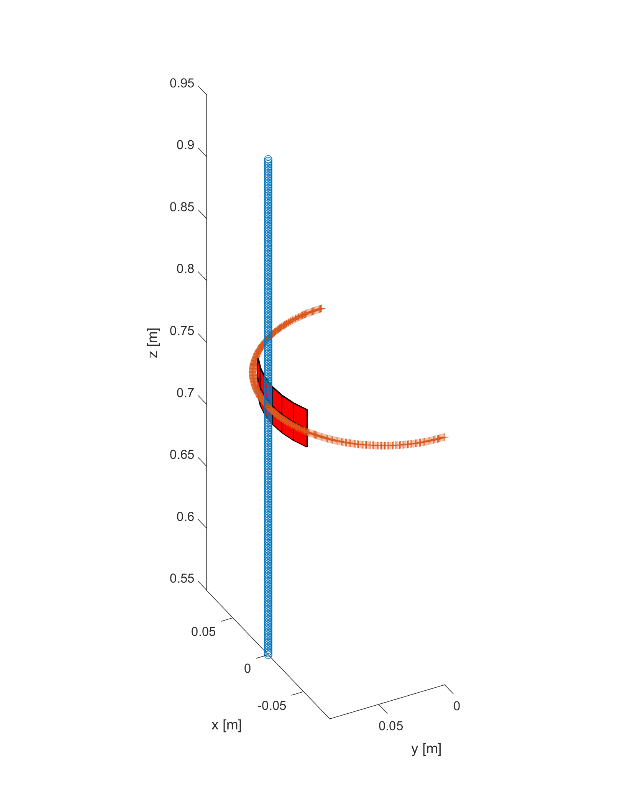}}
    \subfloat[][\emph{Perturbation by flaws}]
    {\includegraphics[width=0.45\linewidth]{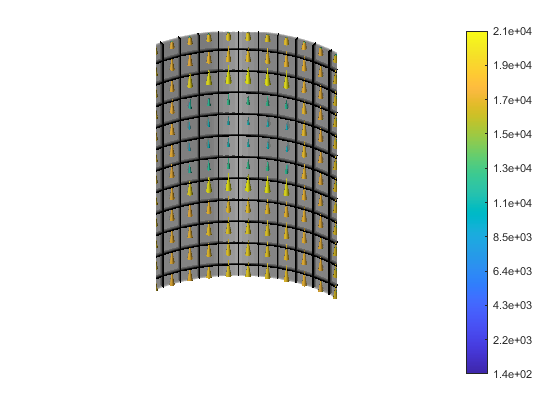}}
    \caption{Left: position of sensing points. Right: zoom showing the density current perturbation due to the crack}    
    \label{fig:scan}
\end{figure}

\subsection{Crack signature definition}\label{sec:postpro}
{The data collected from the sensing points are subsequently processed to obtain useful information about the cracks. The purpose is to characterize the different cracks through an appropriate figure of merit capable of providing a representation of the peculiar characteristics of each. Such figure is called \emph{crack signature} and allows to evaluate the distinctive response of a crack with respect to the measurement points, for a certain excitation frequency. Assuming the $e^{j\omega_k t}$ time behaviour, with $\omega_k=2\pi f_k$, the \emph{crack signature}, for the two crack analyzed, is defined as}
\begin{align}
        \Delta \mathbf{B}_k^{ci}(\varphi_h)&=\mathbf{B}_k^{ci}(\varphi_h)-\mathbf{B}_k^{bg}(\varphi_h) \label{eqn:diff1} \\
        \Delta \mathbf{B}_k^{co}(\varphi_h)&=\mathbf{B}_k^{co}(\varphi_h)-\mathbf{B}_k^{bg}(\varphi_h), \label{eqn:diff2}
\end{align}
{where $\mathbf{B}_k(\varphi_h)$ is the phasor related to the $\varphi$-component of the magnetic flux density at the frequency $f_k$ and in the measurement point described by the angular coordinate $\varphi_h$. The superscripts $bg$, $ci$, $co$ are referred to the background configuration, to the configuration with the inner crack and the configuration with the outer crack, respectively.}

{From the \emph{crack signature}, it is possible to infer some useful information about cracks shape and position. For example, the signals produced by the surface crack have an higher amplitude for higher frequencies even if it can be detected by means of lower frequencies also. On the other hand, useful information on the inner crack can be obtained by using higher frequencies only, due to the penetration depth. For the cracks analyzed in this paper, similar considerations can be drawn from Figure~\ref{fig:sig}, where the behaviour of $\Delta \mathbf{B}_1^{ci}$, $\Delta \mathbf{B}_1^{co}$ and $\Delta \mathbf{B}_{30}^{ci}$, $\Delta \mathbf{B}_{30}^{co}$ with respect to the sensing points is depicted in the complex plane.} 

\begin{figure}[htp]
    \centering
    \subfloat[][\emph{Normalized $\Delta \mathbf{B}_1^{ci}$, $f=1\,\text{Hz}$}]
    {\includegraphics[width=.45\textwidth]{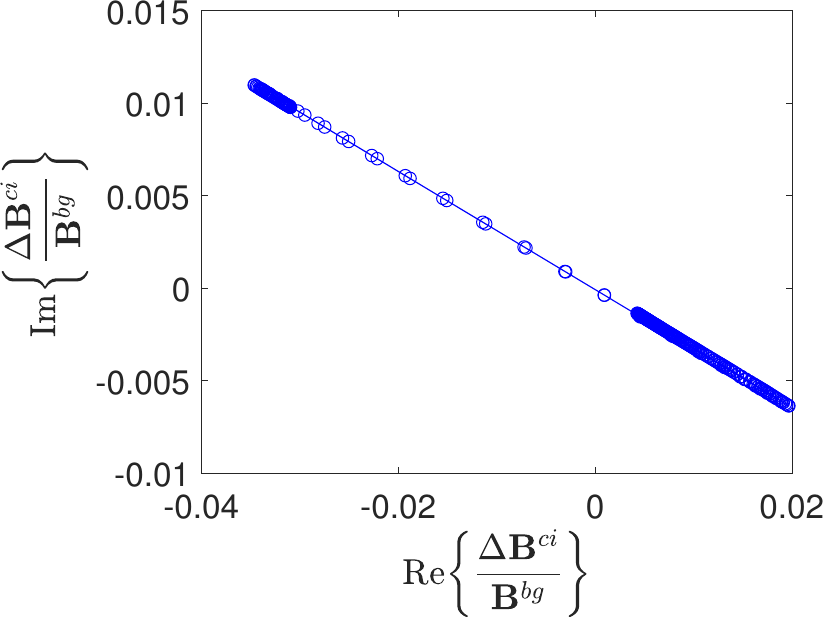}} \quad
    \subfloat[][\emph{Normalized $\Delta \mathbf{B}_{30}^{ci}$, $f=10\,\text{kHz}$}]
    {\includegraphics[width=.45\textwidth]{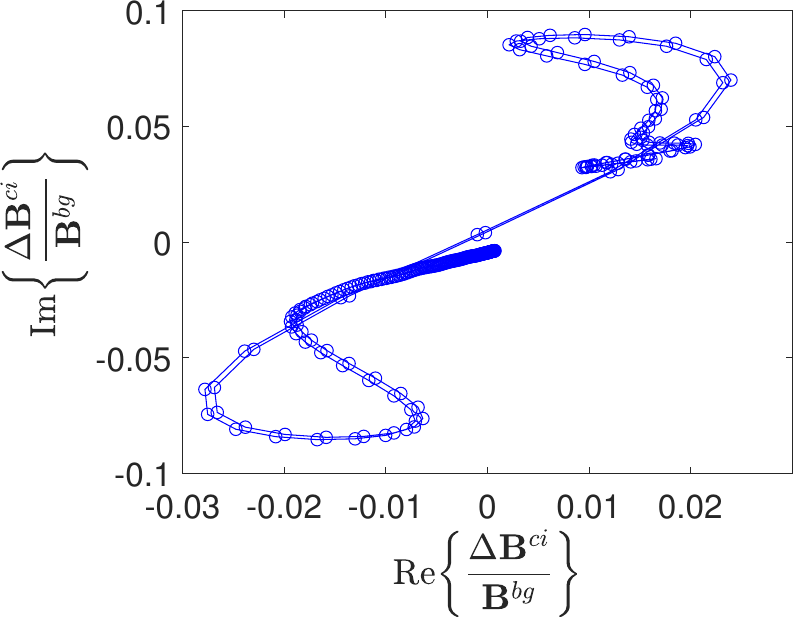}} \\
    \subfloat[][\emph{Normalized $\Delta \mathbf{B}_1^{co}$, $f=1\,\text{Hz}$}]
    {\includegraphics[width=.45\textwidth]{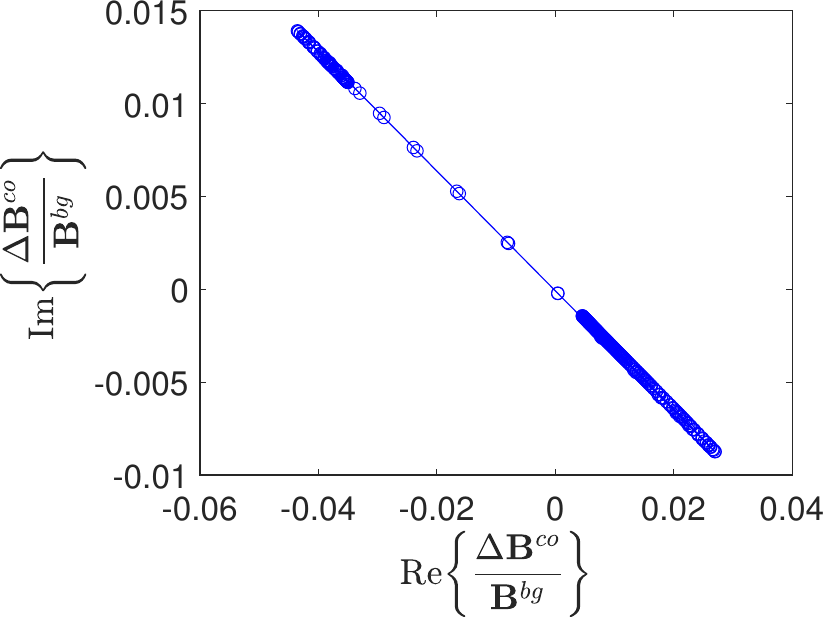}} \quad
    \subfloat[][\emph{Normalized $\Delta \mathbf{B}_{30}^{co}$, $f=10\,\text{kHz}$}]
    {\includegraphics[width=.45\textwidth]{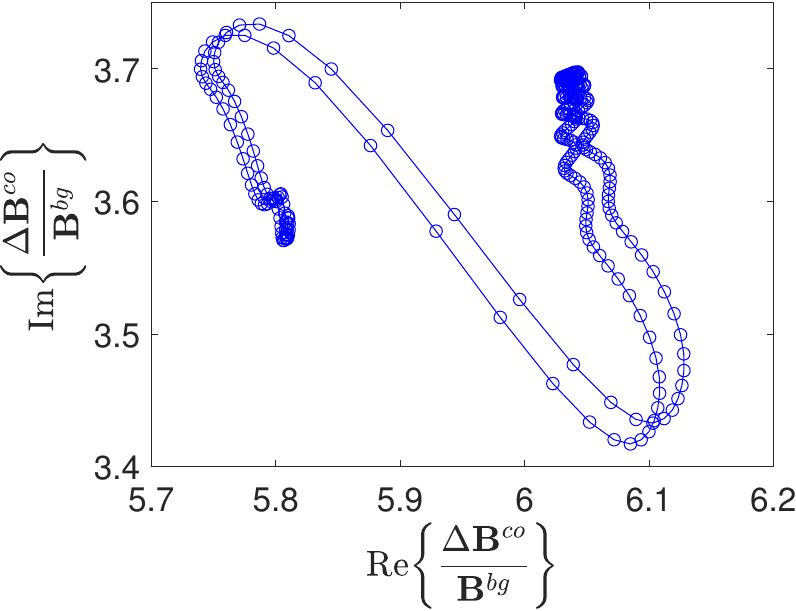}}
    \caption{Characterization of the cracks response. Top: outer crack, bottom: inner crack.}
    \label{fig:sig}
\end{figure}

\subsection{Simulation techniques}
{As underlined in the previous section, the final result of the computation is the quantity $\Delta \mathbf{B}_k$, for the $N_s$ frequencies of interest.}

{From a numerical point of view, the \emph{crack signature} can be evaluated by two different approaches:}
\begin{itemize}
    \item techniques in the frequency domain (FD);
    \item techniques in the time domain (TD).
\end{itemize}

{In the frequency-domain approach, problem~\eqref{weak_form_Ik_a},~\eqref{weak_form_Ik_b} is solved for an excitation $i(t)=\sqrt{2}I_k\sin(2\pi f_k t)$ and a solution of the form $\mathbf{j}(x,t)=\operatorname{Im}\{\mathbf{j_r}(x)e^{j\omega_k t}\}$, where $\omega_k=2\pi f_k$ is the angular frequency. In this case the solution of $N_s$ different problems is required. It is worth noting with this techniques, $N_s$ factorizations of dense matrices are required.}

{In the time-domain approach, the solver computes the solution of problem~\eqref{weak_form_Ik_a},~\eqref{weak_form_Ik_b} for the input current reported in~\eqref{eqn:exc}. The computation result is the time varying behaviour of the magnetic flux density. For the case of interest, a sample time of $30,000$ samples per second is chosen. The phasors of interest have to be subsequently computed and the details are reported below.}

{The signals resulting from the computation are denoted with $B_{\varphi}^{bg}(\varphi,t)$, $B_{\varphi}^{ci}(\varphi,t)$ and $B_{\varphi}^{co}(\varphi,t)$, where $\varphi$ is the angular coordinate in the standard cylindrical coordinates, $B_{\varphi}$ stands for the $\varphi$ component of the magnetic flux density field and the superscripts $bg$, $ci$, $co$ are defined as in Section~\ref{sec:postpro}. The data are subsequently processed as follows}
{
\begin{enumerate}
    \item Evaluating the Fourier transforms $B_{\varphi}^{bg}(\varphi,f)$, $B_{\varphi}^{ci}(\varphi,f)$, $B_{\varphi}^{co}(\varphi,f)$ of $B_{\varphi}^{bg}(\varphi,t)$, $B_{\varphi}^{ci}(\varphi,t)$, $B_{\varphi}^{co}(\varphi,t)$, respectively.
    \item Assembling the phasors $\{\mathbf{B}^{bg}_k(\varphi_h)\}_{h,k}=\lvert B_{\varphi}^{bg}(\varphi_h,f_k)\rvert \exp\left(j\phase{B_{\varphi}^{bg}(\varphi_h,f_k)}\right)$ related to the background signal, for each $h,k$. 
    \item Assembling the phasors $\{\mathbf{B}^{ci}_k(\varphi_h)\}_{h,k}$, $\{\mathbf{B}^{co}_k(\varphi_h)\}_{h,k}$ related to the signals $B_{\varphi}^{ci}$ and $B_{\varphi}^{co}$, respectively.
    \item Computing the \emph{crack signatures} as defined in~\eqref{eqn:diff1},~\eqref{eqn:diff2}.
\end{enumerate}
}

{While we limit the analysis to the standard method for the frequency domain approach (FD), in the time-domain framework, the two different techniques described in previous Section are investigated, namely:}
{
\begin{itemize}
    \item traditional integration scheme based on Cholesky factorization (TD1);
    \item the proposed method based on mode expansion of the field (TD2).
\end{itemize}}

{The numerical examples are carried out by means of in-house software implementing the three methods considered in the comparison. The software is based on the Finite Elements Method and implements parallel computing. Linear algebra computations are based on the usage of SCALAPACK library and, in particular, PDPOTRF is used for the factorization of real matrices (TD1), PDSYGVX for the solution of the generalized eigenvalue problem (TD2) and PZGETRF for the factorization of complex matrices (FD).}

{The mesh employed in domain's discretization has $58332$ nodes and $52940$ hexahedral elements, with a total of $100676$ degrees of freedom (see Figure \ref{fig:mesh}).}
\begin{figure}[htb]
    \centering
    \includegraphics[width=0.25\textwidth]{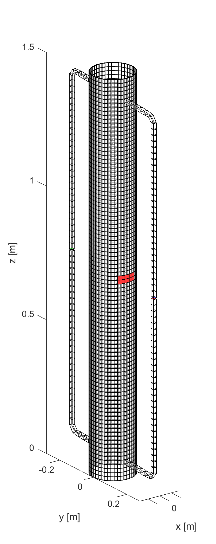}
    \caption{Mesh employed in the calculations.}
    \label{fig:mesh}
\end{figure}

\subsection{Computational time comparison}~\label{sec:comp}
{The simulation times for the problem described in the previous Section with the three mentioned different techniques are reported in Table~\ref{tab:timing}.}

\begin{table}[htb]
    \centering
    \begin{tabular}{c|ccc}
        Method & Setup Time (h) & Transient (h) & Total (h)  \\
        \hline
        FD          & $\sim 40$ & /          & $\sim 40$ \\
        TD1         & 0.5       & $\sim 70$  & $\sim 71$ \\
        TD2         & $\sim 8$  & $\sim 12$  & $\sim 20$ \\        
    \end{tabular}
    \caption{Time comparison between the proposed approach, a traditional integration scheme and the solution in the frequency domain.}
    \label{tab:timing}
\end{table}

{The presented simulation times are achieved utilizing 49 MPI processes on a shared-memory machine. This machine is equipped with dual AMD Epyc 32-Core 7452 processors and boasts 1 TB of RAM.}

{As it can be seen, the proposed approach allows to obtained a relevant gain (about 3.5x) on the total computational time, reducing it from 3 days to 20 hours roughly, with respect to the traditional integration scheme in the time domain. The gain is achieved by a significant decrease of the time needed for the solution of the transient problem, even if the setup time required is much higher with respect to the Cholesky decomposition. Furthermore, the method performs well with respect to the solution in the frequency domain, where the gain is roughly of 2x.}

{In this sense, the proposed approach can outperform traditional methods when the simulation of a large number of time steps is required or a relevant number of different frequencies have to be taken into account. } {Moreover, the longer set-up time is largely compensated when several time domain simulations of the same geometry with different driving terms, and possibly with different time steps, are planned.}

\section{Conclusions}\label{sec5}

In this manuscript we illustrated the use of a modal decomposition approach for the simulation of large time-domain eddy current problems, in presence of injected currents in the conductors through equipotential electrodes. In this respect, we extended the approach of in \cite{tamburrino2021themonotonicity}, where the modal decomposition was introduced in the absence of impressed currents. The key ingredient to get the modal decomposition for this more general case is the identification of two orthogonal subspaces generating $H_L(\Omega,\Sigma)$ as a direct sum. In particular, the space of divergence-free vector fields in $\Omega$ is decomposed into the subspace of purely solenoidal vector fields in $\Omega$ ($H_0$) and its orthogonal subspace ($H_D$). We proved that the finite dimensional subspace $H_D$ has a dimension of $N_e - 1$, where $N_e$ is the number of boundary electrodes.
The current densities in $H_D$ are related to the injected currents across the electrodes, (see (\ref{alphaD}). On the other hand, the infinite discrete eigenmodes of the problem are related to the current distributions in $H_0$, and represent solenoidal current densities circulating into the interior of $\Omega$, the conducting domain (see\ref{eq:generalized_eigenvalues}). 

In Section \ref{sec:discrete_model} we examined the implications of our continuum considerations to finite element approximations of problem (\ref{alphaD})-(\ref{weak_form_J0}). In particular we considered a Galerkin formulation which takes as space of test functions any proper finite vector subspace $X$ of $H_L(\Omega,\Sigma)$. The discrete version of the current injection operator (\ref{eq:operator_electrodes}), resulted in a reduced incidence matrix $\Amat[]{E}$, defined in \ref{eq:definition_E}. Matrix $\Amat[]{E}$ describes Kirchhoff's current laws at the electrodes, in terms of actual degrees of freedom for the eddy current problem. Most important, the image of $\Amat[]{E}^T$ provides a convenient choice for the finite dimensional space $H_D$. Indeed with this choice the matrix $\Amat[]{F}$ introduced in equation (\ref{eq:definition_electrode_incidence_matrix_continuum}) for the original problem coincide with the matrix product $\Amat[]{E}^T \, \Amat[]{E}$. The discrete version of the integral problem (\ref{alphaD})-(\ref{weak_form_J0}) was given in equations (\ref{eq:current_electrodes})-(\ref{eq:discrete_version_2}), and again the eigenmodes of the problem were associated to (a finite subset of) solenoidal currents in the conducting domain, see equation (\ref{eq:modal_decomposition_discrete}).

In Section \ref{sec:results} we verified the efficacy of the modal decomposition approach developed in this manuscript for a Non-destructive-Testing application. In particular we considered a typical pipe used in steam turbines of nuclear power plants, and simulated the crack signature of two possible pipe defects both via the modal decomposition approach and via the usual Cholesky factorization. The determination of generalized eigenvectors and eigenvalues requires much more time than the Cholesky factorization of the problem ($\simeq 8$ h and $0.5$ h respectively). Anyway the solution of the diagonal problem at each time step requires only $\simeq N$ multiplications, \textit{i.e.} much less than the $\simeq N^2$ multiplications required by the forward and backward substitution when solving the problem via the Cholesky-factorized dynamic matrix. In our test case, a wide frequency range of the current drive was necessary to identify the crack features. Hence it was necessary to set up a relatively long simulation compared to the highest frequency we wanted to observe, resulting in a large number of samples. This circumstance determined an overall advantage in using a modal decomposition approach compared to the standard one. NDT physicist and engineers which need to perform long eddy currents time-domain simulations with a relatively small time step, may benefit of the modal decomposition approach discussed in this manuscript. This advantage is even more evident when different simulations of the same geometry with different time steps need to be performed: the generalized eigenmode analysis is independent from the chosen time step, and the inversion of a diagonal matrix requires only $N$ operations. On the other hand the matrices resulting from the Cholesky factorization depend on the time step, and the factorization has to be repeated any time one wants to change the time step of the simulation. 


\section*{Acknowledgments}
The authors wish to thank F. Villone and G. Rubinacci for useful suggestions. 











\nocite{*}
\bibliography{reference_repository_3}%

\clearpage



\end{document}